\documentclass[submission,copyright,creativecommons]{eptcs}
\providecommand{\event}{NCL 2022} 
\usepackage{breakurl}             
\usepackage{underscore}           
\usepackage{amsmath}
\usepackage{amsfonts}
\usepackage{amssymb}
\usepackage{amsthm}
\usepackage{xcolor}

\title{Algebraizability of the Logic of Quasi-N4-Lattices}
\author{Clodomir Silva Lima Neto\thanks{Clodomir Silva Lima Neto thanks the support of the Federal Institute of Education, Science and Technology of Ceará.}
\institute{UFRN \\ Natal, Brazil}
\institute{Department of Informatics and Applied Mathematics \\
Federal University of Rio Grande do Norte \\
Natal, Brazil}
\email{clodomir.neto@ifce.edu.br}
\and
Thiago Nascimento da Silva\thanks{Thiago Nascimento da Silva was financed in part by the Coordenação de Aperfeiçoamento de Pessoal de Nível Superior - Brasil (CAPES) - Finance Code 001.} \qquad \qquad Umberto Rivieccio
\institute{Federal University of Rio Grande do Norte \\ Natal, Brazil}
\email{\quad thiagnascsilva@gmail.com \qquad \qquad umberto.rivieccio@ufrn.br}
}
\def\titlerunning{Algebraizability of the Logic of Quasi-N4-Lattices}
\def\authorrunning{C. S. Lima Neto, T. N. da Silva \& U. Rivieccio}

\newcommand{\nnot}{\mathop{\sim}}
\newcommand{\eq}{\leftrightarrow}

\theoremstyle{definition}
\newtheorem{dfn}{Definition}
\newtheorem{lem}{Lemma}
\newtheorem{prop}{Proposition}
\newtheorem{thm}{Theorem}
\newtheorem{corollary}{Corollary}

\begin{document}
\maketitle

\begin{abstract}
The class of quasi-N4-lattices (QN4-lattices) was introduced as a common generalization of quasi-Nelson algebras and N4-lattices, in such a way that N4-lattices are precisely the QN4-lattices satisfying the double negation law $(\nnot \nnot x = x)$ and quasi-Nelson algebras are the QN4-lattices satisfying the explosive law $(x \land \nnot x) \to y = ((x \land \nnot x) \to y) \to ((x \land \nnot x) \to y)$. In this paper we introduce,  via a Hilbert-style presentation, a logic $(\mathbf{L}_{\mathrm{QN4}})$ whose algebraic semantics is a class of algebras that we show to be term-equivalent to  QN4-lattices. The result is obtained by showing that the calculus introduced by us is algebraizable in the sense of Blok and Pigozzi, and its equivalent algebraic semantics is term-equivalent to  the class of QN4-lattices. As a prospect for future investigation, we  consider the question of how one could place $\mathbf{L}_{\mathrm{QN4}}$ within the family of relevance logics.
\end{abstract}


\section{Introduction}
Nelson's constructive logic with strong negation  (N3), introduced in \cite{Nelson1949}, is a conservative expansion of the negation-free fragment of intuitionistic propositional logic by an unary logical connective $\nnot$ of strong negation (which is involutive and satisfies De Morgan's laws). Within the Nelson family, two mutually incomparable generalizations of N3 have been proposed so far, namely: (i) paraconsistent Nelson's logic (N4), which is a paraconsistent weakening of N3 (see~\cite{Almukdad1984}) that results from dropping the explosion axiom; (ii) the logic (dubbed quasi-Nelson logic) obtained from N3 by deleting the double negation law, which is also weaker than intuitionistic logic \cite{Rivieccio2018}.

The algebraic counterpart of N3 is the variety of  Nelson algebras. Rivieccio and Spinks \cite{Rivieccio2018} introduced quasi-Nelson algebras as a natural generalization of Nelson algebras in the sense that the negation $\nnot$ need not be involutive. The algebraic counterparts of paraconsistent Nelson's logic N4 and quasi-Nelson logic are, respectively, the variety  of N4-lattices and the variety of quasi-Nelson algebras. More recently, Rivieccio \cite{Rivieccio2022} introduced the class of quasi-N4-lattices (QN4-lattices) as a common generali\-zation of these two varieties.

The definition of QN4-lattices is such that N4-lattices turn out to be precisely the quasi-N4-lattices satisfying the double negation law, and quasi-Nelson algebras are precisely the QN4-lattices satisfying the explosive law. 

The language of quasi-N4-lattices includes two implication connectives, the strong implication $(\Rightarrow)$ that forms a residuated pair together with the strong conjunction ($*$), and the weak implication $(\to)$ that enjoys the standard version of the Deduction Theorem; the former is definable from the latter (and the lattice meet) by $x \Rightarrow y := (x \to y) \to (\nnot y \to \nnot x)$. It is the strong implication that determines the lattice order on a quasi-N4-lattice; it follows that an equivalence connective can be defined by $x \Leftrightarrow y := (x \Rightarrow y) \land (y \Rightarrow x)$.

Nelson algebras, quasi-Nelson algebras and N4-lattices can be represented via a construction known as twist-structure. As shown in~\cite{Rivieccio2022}, the class of QN4-lattices also admits a twist-structure representation  analogous to the above-mentioned ones, though the  construction needs to be generalized to account for both the non-involutivity of the negation and the lack of the explosive law. To accomplish this, this representation employs twist-structures defined over Brouwerian algebras enriched with a nucleus operator. 

In this contribution we are going to introduce a logic $(\mathbf{L}_{\mathrm{QN4}})$ via a Hilbert-style and show that $\mathbf{L}_{\mathrm{QN4}}$ is algebraizable in the sense of Blok and Pigozzi. We will then prove that the equivalent algebraic semantics of $\mathbf{L}_{\mathrm{QN4}}$ is term-equivalent to the class of QN4-lattices.

The paper is organized as follows. In Section \ref{preliminaries} we recall some basic definitions and results about quasi-N4-lattices. Section \ref{calculus} introduces $\mathbf{L}_{\mathrm{QN4}}$ through  a Hilbert-style calculus. In Section \ref{bpa} we prove that $\mathbf{L}_{\mathrm{QN4}}$ is algebraizable. In Section \ref{semantics} we show that the algebraic counterpart of $\mathbf{L}_{\mathrm{QN4}}$ is term-equivalent to the class of quasi-N4-lattices. In the final Section \ref{future} we mention some prospects for future work.

\section{Preliminaries} 
\label{preliminaries}
In this section we recall two equivalent presentations of quasi-N4-lattices; these will be used to establish the equivalence between the two alternative algebraic semantics for the logic $\mathbf{L}_{\mathrm{QN4}}$, which is introduced in the next section. 

\begin{dfn}
A Brouwerian algebra is an algebra $\mathbf{B} = \langle B; \land, \lor, \to \rangle$ such that $\langle B; \land, \lor \rangle$ is a lattice with order $\leq$ and $\to$ is the residuum of $\land$, that is, $a \land b \leq c$ iff $a \leq b \to c$, for all $a, b, c \in B$.
\end{dfn}

As is well known, Brouwerian algebras are precisely the bottom-free subreducts of Heyting algebras -- the algebraic counterpart of intuitionistic logic. 

\begin{dfn}(\cite{Rivieccio2022}, Definition 2.1)
Given a Brouwerian algebra $\mathbf{B} = \langle B; \land, \lor, \to \rangle$, we say that a unary operator $\Box: B \to B$ is a nucleus if, for all $a, b \in B$,
\begin{enumerate}
\item $\Box (a \land b) = \Box a \land \Box b$.
\item $a \leq \Box a = \Box \Box a$.
\end{enumerate}
\end{dfn}

We shall refer to an algebra $\mathbf{B} = \langle B; \land, \lor, \to, \Box \rangle$ as to a nuclear Brouwerian algebra. 

\begin{dfn}(\cite{Rivieccio2022}, Definition 2.2)\label{twist}
Let $\mathbf{B} = \langle B; \land, \lor, \to, \Box \rangle$ be a nuclear Brouwerian algebra. The algebra $\mathbf{B}^{\Join} = \langle B \times B; \land, \lor, \to, \nnot \rangle$ is defined as follows. For all $\langle a_{1}, a_{2} \rangle,  \langle b_{1}, b_{2} \rangle \in B \times B$,
\begin{align*}
\nnot \langle a_{1}, a_{2} \rangle &= \langle a_{2}, \Box a_{1} \rangle \\
\langle a_{1}, a_{2} \rangle \land \langle b_{1}, b_{2} \rangle &= \langle a_{1} \land b_{1} , \Box (a_{2} \lor b_{2}) \rangle  \\
\langle a_{1}, a_{2} \rangle \lor \langle b_{1}, b_{2} \rangle &= \langle a_{1} \lor b_{1} , a_{2} \land b_{2} \rangle \\
\langle a_{1}, a_{2} \rangle \to \langle b_{1}, b_{2} \rangle &= \langle a_{1} \to b_{1} , \Box a_{1} \land b_{2} \rangle
\end{align*}
A quasi-N4 twist-structure $\mathbf{A}$ over $\mathbf{B}$ is a subalgebra of $\mathbf{B}^{\Join}$ satisfying the following properties: $\pi_{1}[A] = B$ and $\Box a_{2} = a_{2}$ for all $\langle a_{1}, a_{2} \rangle \in A$, where $\pi_{1}$ denote the first projection function.
\end{dfn}

Given an algebra $\mathbf{A}$ having an operation $\to$ and elements $a, b \in A$, we shall abbreviate $|a| := a \to a$, and define the relations $\equiv$ and $\preceq$ as follows. We let $a \preceq b$ iff $a \to b = |a \to b|$, and $\equiv := \preceq \cap (\preceq)^{-1}$. Thus one has $a \equiv b$ iff ($a \preceq b$ and $b \preceq a$).

\begin{dfn}(\cite{Rivieccio2022}, Definition 3.2) \label{qn4lattice}
A quasi-N4-lattice (QN4-lattice) is an algebra $\mathbf{A} = \langle A; \land, \lor, \to, \nnot \rangle$ of type $\langle 2, 2, 2, 1 \rangle$ satisfying the following properties:
\begin{description}
\item[(QN4a)] The reduct $\langle A; \land, \lor \rangle$ is a distributive lattice with lattice order $\leq$.
\item[(QN4b)] The relation $\equiv := \preceq \cap \, (\preceq)^{-1}$ is a congruence on the reduct $\langle A; \land, \lor, \to \rangle$ and the quotient $B(\mathbf{A}) = \langle A; \land, \lor, \to \rangle / \equiv$ is a Brouwerian algebra. The operator $\Box$ given by $\Box[a] := \nnot \nnot a / \equiv$ for all $a \in A$ is a nucleus, so the algebra $\langle B(\mathbf{A}), \Box \rangle$ is a nuclear Brouwerian algebra.
\item[(QN4c)] For all $a, b \in A$, it holds that $a \leq b$ iff $a \preceq b$ and $\nnot b \preceq \nnot a$.
\item[(QN4d)] For all $a, b \in A$, it holds that $\nnot (a \to b) \equiv \nnot \nnot (a \, \land \nnot b)$.
\item[(QN4e)] For all $a, b \in A$,
\begin{description}
\item[(QN4e.1)] $a \leq \nnot \nnot a$.
\item[(QN4e.2)] $\nnot a = \nnot \nnot \nnot a$.
\item[(QN4e.3)] $\nnot (a \lor b) = \nnot a \, \land \nnot b$.
\item[(QN4e.4)] $\nnot \nnot a \land \nnot \nnot b = \nnot \nnot (a \land b)$.
\end{description}
\end{description}
\label{qn4l}
\end{dfn}

The preceding definition is a straightforward generalization of Odintsov's~\cite{Odintsov2003} definition of N4-lattices; indeed, as observed in~\cite[Proposition 3.8]{Rivieccio2022}, a quasi-N4-lattice  $\mathbf{A}$ is an N4-lattice if and only if $\mathbf{A}$ is involutive, that is, 
$\nnot \nnot a \leq a $ for all $a \in A$. Similarly, a quasi-Nelson algebra
may be defined as a quasi-N4-lattice $\mathbf{A}$ that satisfies the explosive equality, $a \land \nnot a \preceq b $ for all $a,b \in A$. 

\begin{thm}(\cite{Rivieccio2022}, Theorem 3.3)
Every quasi-N4-lattice $\mathbf{A}$ is isomorphic to a twist-structure over $\langle B(\mathbf{A}), \Box \rangle$ by the map $\iota: A \to A / \equiv \times A / \equiv$ given by $\iota(a) := \langle a / \equiv, \nnot a / \equiv \rangle$ for all $a \in A$.
\label{thmrep}
\end{thm}

In the proposition below we see that the non-equational presentation for QN4-lattices given in Definition \ref{qn4l} can be replaced with an equational one, entailing that QN4-lattices form a variety of algebras.

\begin{prop}(\cite{Rivieccio2022}, Proposition 3.7) Items \textbf{(QN4b)} and \textbf{(QN4c)} in Definition \ref{qn4l} can be equivalently replaced by the following identities:
\begin{enumerate}
\item $| x | \to y \approx y$.
\item $(x \land y) \to x \approx | (x \land y) \to x |$.
\item $(x \land y) \to z \approx x \to (y \to z)$.
\item $(x \Leftrightarrow y) \to x \approx (x \Leftrightarrow y) \to y$.
\item $(x \lor y) \to z \approx (x \to z) \land (y \to z)$.
\item $x \to (y \land z) \approx (x \to y) \land (x \to z)$.
\item $(x \to y) \land (y \to z) \preceq x \to z$.
\item $x \to y \preceq x \to (y \lor z)$.
\item $x \to (y \to z) \approx (x \to y) \to (x \to z)$.
\item $x \to y \preceq \nnot \nnot x \to \nnot \nnot y$.
\end{enumerate}
\label{variety}
\end{prop}

\section{A Hilbert calculus for $\mathbf{L}_{\mathrm{QN4}}$}
\label{calculus}

In this section we introduce a Hilbert-style calculus that determines a logic, in sense of \cite{Hamilton1998}, henceforth denoted by $\mathbf{L}_{\mathrm{QN4}}$. Moreover, we highlight some theorems and derivations of $\mathbf{L}_{\mathrm{QN4}}$ that will be used to prove its algebraizability in subsequent sections.

Fix a denumerable set $\mathcal{P}$ of propositional variables, and let $p \in \mathcal{P}$. The language $\mathcal{L}$ of QN4-lattice over $\mathcal{P}$ is defined recursively as follows:
$$
\alpha ::= p \mid \, \nnot \alpha \mid (\alpha \land \alpha) \mid (\alpha \lor \alpha) \mid (\alpha \to \alpha)
$$

To simplify the notation, in what follows, we omit the outmost parenthesis. We also abbreviate $\alpha \leftrightarrow \beta := (\alpha \to \beta) \land (\beta \to \alpha)$. We use $\mathit{F_{\mathcal{P}}}$ to denote the set of all formulas. A \textit{logic} is then defined as a finitary and substitution-invariant consequence relation $\vdash \subseteq \wp (\mathit{F_{\mathcal{P}}}) \times \mathit{F_{\mathcal{P}}}$. The Hilbert-system for $\mathbf{L}_{\mathrm{QN4}}$ consists of the following axiom schemes together with the single inference rule of \textit{modus ponens} (MP): $\alpha, \alpha \to \beta \vdash \beta$.

\begin{description}
\item[$\mathbf{Ax1}$]  $\alpha \to (\beta \to \alpha)$
\item[$\mathbf{Ax2}$] $(\alpha \to (\beta \to \gamma)) \to ((\alpha \to \beta) \to (\alpha \to \gamma))$
\item[$\mathbf{Ax3}$] $(\alpha \land \beta) \to \alpha$
\item[$\mathbf{Ax4}$] $(\alpha \land \beta) \to \beta$
\item[$\mathbf{Ax5}$] $(\alpha \to \beta) \to ((\alpha \to \gamma) \to (\alpha \to (\beta \land \gamma)))$
\item[$\mathbf{Ax6}$] $\alpha \to (\alpha \lor \beta)$
\item[$\mathbf{Ax7}$] $\beta \to (\alpha \lor \beta)$
\item[$\mathbf{Ax8}$] $(\alpha \to \gamma) \to ((\beta \to \gamma) \to ((\alpha \lor \beta) \to \gamma))$
\item[$\mathbf{Ax9}$] $\nnot(\alpha \lor \beta) \eq (\nnot \alpha \land \nnot \beta)$
\item[$\mathbf{Ax10}$] $\nnot (\alpha \to \beta) \eq \nnot \nnot (\alpha \land \nnot \beta)$
\item[$\mathbf{Ax11}$] $\nnot(\alpha \land (\beta \land \gamma)) \eq \nnot((\alpha \land \beta) \land \gamma)$
\item[$\mathbf{Ax12}$] $\nnot(\alpha \land (\beta \lor \gamma)) \eq \nnot((\alpha \land \beta) \lor (\alpha \land \gamma))$
\item[$\mathbf{Ax13}$] $\nnot(\alpha \lor (\beta \land \gamma)) \eq \nnot((\alpha \lor \beta) \land (\alpha \lor \gamma))$
\item[$\mathbf{Ax14}$] $\nnot \nnot(\alpha \land \beta) \eq (\nnot \nnot \alpha \land \nnot \nnot \beta)$
\item[$\mathbf{Ax15}$] $\alpha \to \nnot \nnot \alpha$
\item[$\mathbf{Ax16}$] $\alpha \to (\nnot \alpha \to \nnot (\alpha \to \alpha))$
\item[$\mathbf{Ax17}$] $(\alpha \to \beta) \to (\nnot \nnot \alpha \to \nnot \nnot \beta)$
\item[$\mathbf{Ax18}$] $\nnot \alpha \to \nnot(\alpha \land \beta)$
\item[$\mathbf{Ax19}$] $\nnot(\alpha \land \beta) \to \nnot(\beta \land \alpha)$
\item[$\mathbf{Ax20}$] $(\nnot \alpha \to \nnot \beta) \to (\nnot (\alpha \land \beta) \to \nnot \beta)$	
\item[$\mathbf{Ax21}$] $(\nnot \alpha \to \nnot \beta) \to ((\nnot \gamma \to \nnot \theta) \to (\nnot(\alpha \land \gamma) \to \nnot(\beta \land \theta)))$
\item[$\mathbf{Ax22}$] $\nnot \nnot \nnot \alpha \to \nnot \alpha$
\end{description}

Axioms $\mathbf{Ax1}$-$\mathbf{Ax8}$ together with modus ponens constitute an axiomatization of the Positive Logic ($\mathbf{Lp}$). We started with them and choose between the axioms of quasi-Nelson logic \cite{Liang2019} the ones that were sound with respect to QN4-lattices, then we added remaining axioms necessary to prove that our calculus is algebraizable and that its equivalent algebraic semantics is the class of QN4-lattices as defined in Definition \ref{qn4lattice}.

By the usual inductive argument on the length of derivations, it is not difficult to prove that the deduction theorem holds for $\mathbf{L}_{\mathrm{QN4}}$.

\begin{thm}
\textbf{(Deduction Theorem)} If $\Phi \cup \{ \alpha \} \vdash_{\mathbf{L}_{\mathrm{QN4}}} \beta$, then $\Phi \vdash_{\mathbf{L}_{\mathrm{QN4}}} \alpha \to \beta$.
\end{thm}

The following lemma is an immediate consequence of the Deduction Theorem. 

\begin{lem}
If $\alpha, \beta, \gamma \in \mathcal{L}$ then 
\begin{enumerate}
\item $\vdash_{\mathbf{L}_{\mathrm{QN4}}} \alpha \to \alpha$.
\item $\{ \alpha \to \beta, \beta \to \gamma \} \vdash_{\mathbf{L}_{\mathrm{QN4}}} \alpha \to \gamma$.
\end{enumerate}
\label{lema}
\end{lem}

\section{$\mathbf{L}_{\mathrm{QN4}}$ is BP-Algebraizable}
\label{bpa}
In this section we prove that the calculus introduced in the previous section is algebraizable in sense of Blok and Pigozzi. Using this result, we will axiomatize the equivalent algebraic semantics of $\mathbf{L}_{\mathrm{QN4}}$ via the algorithm of (\cite{Blok2014}, Theorem 2.17) and  show that is term-equivalent to the class of QN4-lattices.

Given the formula algebra \textbf{Fm}, the associated set of equations of the language $\mathcal{L}$ is denoted by $\mathit{Eq}$ and is defined as $\mathit{Eq} := \mathit{F_m} \times \mathit{F_m}$. Following standard usage, we denote an equation $(\alpha, \beta)$ as $\alpha \approx \beta$.

\begin{thm}
A logic $\mathbf{L}$ is algebraizable if and only if there are a set of equations $E(\alpha) \subseteq \mathit{Eq}$ and a set of formulas $\Delta(\alpha, \beta) \subseteq \mathit{F_m}$, such that  the following conditions hold:
\begin{description}
\item[$(\mathbf{Ref})$] $\vdash_{\mathbf{L}} \Delta(\alpha, \alpha)$
\item[$(\mathbf{Sym})$] $\Delta(\alpha, \beta) \vdash_{\mathbf{L}} \Delta(\beta, \alpha)$
\item[$(\mathbf{Trans})$] $\Delta(\alpha, \beta)  \cup \Delta(\beta, \gamma) \vdash_{\mathbf{L}} \Delta(\alpha, \gamma)$
\item[$(\mathbf{Alg})$] $\alpha \dashv \vdash_{\mathbf{L}} \Delta(E(\alpha))$
\item[$(\mathbf{Cong})$] for each n-ary connective $\bullet$, $\displaystyle \bigcup^{n}_{i =1} \Delta(\alpha_i, \beta_i) \vdash_{\mathbf{L}} \Delta(\bullet(\alpha_1, \ldots, \alpha_n), \bullet(\beta_1, \ldots, \beta_n))$.
\end{description}
\label{algebraizable}
\end{thm}

As is well known, the conditions $(\mathbf{Sym})$ e $(\mathbf{Trans})$ of in  Theorem \ref{algebraizable} can be replaced by condition $(\mathbf{MP}): \alpha, \Delta(\alpha, \beta) \vdash_{\mathbf{L}} \beta$.

We are going to see that 
$$E(\alpha) := \{ \alpha \approx \alpha \to \alpha \}$$
and 
$$
\Delta(\alpha, \beta) := \{\alpha \to \beta, \beta \to \alpha, \nnot \alpha \to \nnot \beta, \nnot \beta \to \nnot \alpha\}
$$
are, respectively, a set of defining equations and a set of equivalence formulas that witness the algebrai\-zability of $\mathbf{L}_{\mathrm{QN4}}$.

For an algebraizable logic $\mathbf{L}$, we say $\mathbf{L}$ is finitely algebraizable when the set of equivalence formulas is finite, and we say $\mathbf{L}$ is BP-algebraizable when it is finitely algebraizable and the set of defining equations is finite.

\begin{thm}
$\mathbf{L}_{\mathrm{QN4}}$ is BP-algebraizable.
\end{thm}

\begin{proof}
In order to prove (\textbf{Ref}), it is necessary to show that $\vdash_{\mathbf{L}_{\mathrm{QN4}}} \{ \alpha \to \alpha, \nnot \alpha \to \nnot \alpha \}$, and it is Lemma \ref{lema}.1. (\textbf{MP}): $\alpha, \{\alpha \to \beta, \beta \to \alpha, \nnot \alpha \to \nnot \beta, \nnot \beta \to \nnot \alpha \} \vdash_{\mathbf{L}_{\mathrm{QN4}}} \beta$ is a straightforward consequence of modus ponens. As to (\textbf{Alg}), it suffices to prove that $\alpha \dashv \vdash_{\mathbf{L}_{\mathrm{QN4}}} \{\alpha \to (\alpha \to \alpha), (\alpha \to \alpha) \to \alpha, \nnot \alpha \to \nnot(\alpha \to \alpha), \nnot (\alpha \to \alpha) \to \nnot \alpha\}$. From right to left, thanks to Lemma \ref{lema}.1 and using MP, we infer the desired result. From left to right, we will prove that: (i) $\alpha \vdash_{\mathbf{L}_{\mathrm{QN4}}} \alpha \to (\alpha \to \alpha)$, we have it by instantiating \textbf{Ax1}; (ii) $\alpha \vdash_{\mathbf{L}_{\mathrm{QN4}}} (\alpha \to \alpha) \to \alpha$, follows from \textbf{Ax1} and MP; (iii) $\alpha \vdash_{\mathbf{L}_{\mathrm{QN4}}} \nnot \alpha \to \nnot (\alpha \to \alpha)$ is logical consequence of \textbf{Ax16} and modus ponens; (iv) $\alpha \vdash_{\mathbf{L}_{\mathrm{QN4}}} \nnot (\alpha \to \alpha) \to \nnot \alpha$, we have
\begin{table}[!ht]
\centering
\begin{tabular}{ll}
1. $\alpha$ & Premise \\ 
2. $\nnot (\alpha \to \alpha) \to \nnot \nnot (\alpha \land \nnot \alpha)$ & Ax10 ($\to$) \\
3. $\nnot \nnot (\alpha \land \nnot \alpha) \to (\nnot \nnot \alpha \land \nnot \nnot \nnot \alpha)$ & Ax14 ($\to$) \\
4. $\nnot (\alpha \to \alpha) \to (\nnot \nnot \alpha \land \nnot \nnot \nnot \alpha)$ & Lemma 1.2, 2, 3 \\
5. $(\nnot \nnot \alpha \land \nnot \nnot \nnot \alpha) \to \nnot \nnot \nnot \alpha$ & Ax4 \\
6. $\nnot (\alpha \to \alpha) \to \nnot \nnot \nnot \alpha$ & Lemma 1.2, 4, 5 \\
7. $\nnot \nnot \nnot \alpha \to \nnot \alpha$ & Ax22 \\
8. $\nnot (\alpha \to \alpha) \to \nnot \alpha$ & Lemma 1.2, 6, 7
\end{tabular} 
\end{table} 

As to (\textbf{Cong}), we need to prove for each connective $\bullet$  $\in \{ \nnot, \land, \lor, \to \}$.

For $(\nnot)$, we need to prove that:
\begin{align}
\{ \alpha \to \beta, \beta \to \alpha, \nnot \alpha \to \nnot \beta, \nnot \beta \to \nnot \alpha \} &\vdash_{\mathbf{L}_{\mathrm{QN4}}} \nnot \alpha \to \nnot \beta \\
\{ \alpha \to \beta, \beta \to \alpha, \nnot \alpha \to \nnot \beta, \nnot \beta \to \nnot \alpha \} &\vdash_{\mathbf{L}_{\mathrm{QN4}}} \nnot \beta \to \nnot \alpha \\
\{ \alpha \to \beta, \beta \to \alpha, \nnot \alpha \to \nnot \beta, \nnot \beta \to \nnot \alpha \} &\vdash_{\mathbf{L}_{\mathrm{QN4}}} \nnot \nnot \alpha \to \nnot \nnot \beta \\
\{ \alpha \to \beta, \beta \to \alpha, \nnot \alpha \to \nnot \beta, \nnot \beta \to \nnot \alpha \} &\vdash_{\mathbf{L}_{\mathrm{QN4}}} \nnot \nnot \beta \to \nnot \nnot \alpha
\end{align}

In (1) and (2), the conclusion follows directly from the premises. Also, in (3) and (4), the conclusion can be inferred from \textbf{Ax17} and MP.

Now consider the following sets, $\Gamma_{1} = \{ \alpha_{1} \to \beta_{1}, \beta_{1} \to \alpha_{1}, \nnot \alpha_{1} \to \nnot \beta_{1}, \nnot \beta_{1} \to \nnot \alpha_{1} \}$ and $\Gamma_{2} = \{ \alpha_{2} \to \beta_{2}, \beta_{2} \to \alpha_{2}, \nnot \alpha_{2} \to \nnot \beta_{2}, \nnot \beta_{2} \to \nnot \alpha_{2} \}$. 

For $(\land)$, we need to prove that:
\begin{align}
\Gamma_{1} \cup \Gamma_{2} &\vdash (\alpha_{1} \land \alpha_{2}) \to (\beta_{1} \land \beta_{2}) \\
\Gamma_{1} \cup \Gamma_{2} &\vdash (\beta_{1} \land \beta_{2}) \to (\alpha_{1} \land \alpha_{2}) \\
\Gamma_{1} \cup \Gamma_{2} &\vdash \nnot (\alpha_{1} \land \alpha_{2}) \to \nnot (\beta_{1} \land \beta_{2}) \\
\Gamma_{1} \cup \Gamma_{2} &\vdash \nnot (\beta_{1} \land \beta_{2}) \to \nnot (\alpha_{1} \land \alpha_{2})
\end{align}

The same reasoning from (5) will be used in (6), so we will only show item (5), see next page.

\begin{table}[!ht]
\centering
\begin{tabular}{ll}
1. $\alpha_{1} \to \beta_{1}$ & Premise \\ 
2. $\alpha_{2} \to \beta_{2}$ & Premise \\ 
3. $(\alpha_{1} \land \alpha_{2}) \to \alpha_{1}$ & Ax3 \\
4. $(\alpha_{1} \land \alpha_{2}) \to \beta_{1}$ & Lemma 1.2, 1, 3 \\
5. $(\alpha_{1} \land \alpha_{2}) \to \alpha_{2}$ & Ax4 \\
6. $(\alpha_{1} \land \alpha_{2}) \to \beta_{2}$ & Lemma 1.2, 2, 5 \\
7. $((\alpha_{1} \land \alpha_{2}) \to \beta_{1}) \to (((\alpha_{1} \land \alpha_{2}) \to \beta_{2}) \to ((\alpha_{1} \land \alpha_{2}) \to (\beta_{1} \land \beta_{2})))$ & Ax5 \\
8. $((\alpha_{1} \land \alpha_{2}) \to \beta_{2}) \to ((\alpha_{1} \land \alpha_{2}) \to (\beta_{1} \land \beta_{2}))$ & MP, 4, 7 \\
9. $(\alpha_{1} \land \alpha_{2}) \to (\beta_{1} \land \beta_{2})$ & MP, 6, 8
\end{tabular} 
\end{table} 

\newpage

The derivation of (7) and (8) are straightforward and make use of \textbf{Ax21} and MP.

For $(\lor)$, we need to prove that:
\begin{align}
\Gamma_{1} \cup \Gamma_{2} &\vdash (\alpha_{1} \lor \alpha_{2}) \to (\beta_{1} \lor \beta_{2}) \\
\Gamma_{1} \cup \Gamma_{2} &\vdash (\beta_{1} \lor \beta_{2}) \to (\alpha_{1} \lor \alpha_{2}) \\
\Gamma_{1} \cup \Gamma_{2} &\vdash \nnot (\alpha_{1} \lor \alpha_{2}) \to \nnot (\beta_{1} \lor \beta_{2}) \\
\Gamma_{1} \cup \Gamma_{2} &\vdash \nnot (\beta_{1} \lor \beta_{2}) \to \nnot (\alpha_{1} \lor \alpha_{2})
\end{align}

For (9) and (10), we use \textbf{Ax6}, \textbf{Ax7}, \textbf{Ax8} and MP for inferring the conclusions. The same reasoning from (11) will be used in (12), so we will only show item (11),

\begin{table}[!ht]
\centering
\begin{tabular}{ll}
1. $\nnot \alpha_{1} \to \nnot \beta_{1}$ & Premise \\
2. $\nnot \alpha_{2} \to \nnot \beta_{2}$ & Premise \\
3. $\overbrace{(\nnot \alpha_{1} \land \nnot \alpha_{2})}^{\varphi} \to \nnot \alpha_{1}$ & Ax3 \\
4. $(\nnot \alpha_{1} \land \nnot \alpha_{2}) \to \nnot \beta_{1}$ & Lemma 1.2, 1, 3 \\
5. $(\nnot \alpha_{1} \land \nnot \alpha_{2}) \to \nnot \alpha_{2}$ & Ax4 \\
6. $(\nnot \alpha_{1} \land \nnot \alpha_{2}) \to \nnot \beta_{2}$ & Lemma 1.2, 2, 5 \\
7. $(\varphi \to \nnot \beta_{1}) \to ((\varphi \to \nnot \beta_{2}) \to (\varphi \to (\nnot \beta_{1} \land \nnot \beta_{2})))$ & Ax5 \\
8. $(\varphi \to \nnot \beta_{2}) \to (\varphi \to (\nnot \beta_{1} \land \nnot \beta_{2}))$ & MP, 4, 7 \\
9. $(\nnot \alpha_{1} \land \nnot \alpha_{2}) \to (\nnot \beta_{1} \land \nnot \beta_{2})$ & MP, 6, 8 \\
10. $\nnot (\alpha_{1} \lor \alpha_{2}) \to (\nnot \alpha_{1} \land \nnot \alpha_{2})$ & Ax9 ($\to$) \\ 
11. $\nnot (\alpha_{1} \lor \alpha_{2}) \to (\nnot \beta_{1} \land \nnot \beta_{2})$ & Lemma 1.2, 9, 10 \\
12. $(\nnot \beta_{1} \land \nnot \beta_{2}) \to \nnot (\beta_{1} \lor \beta_{2})$ & Ax9 ($\leftarrow$) \\
13. $\nnot (\alpha_{1} \lor \alpha_{2}) \to \nnot (\beta_{1} \lor \beta_{2})$ & Lemma 1.2, 10, 11
\end{tabular} 
\end{table}

For $(\to)$, we need to prove that:
\begin{align}
\Gamma_{1} \cup \Gamma_{2} &\vdash (\alpha_{1} \to \alpha_{2}) \to (\beta_{1} \to \beta_{2}) \\
\Gamma_{1} \cup \Gamma_{2} &\vdash (\beta_{1} \to \beta_{2}) \to (\alpha_{1} \to \alpha_{2}) \\
\Gamma_{1} \cup \Gamma_{2} &\vdash \nnot (\alpha_{1} \to \alpha_{2}) \to \nnot (\beta_{1} \to \beta_{2}) \\
\Gamma_{1} \cup \Gamma_{2} &\vdash \nnot (\beta_{1} \to \beta_{2}) \to \nnot (\alpha_{1} \to \alpha_{2})
\end{align}

Lemma \ref{lema}.2 is used in (13) and (14) for inferring the conclusions. The same reasoning from (15) will be used in (16), so we will only show item (15), see next page.

\begin{table}[h]
\centering
\begin{tabular}{ll}
1. $\alpha_{1} \to \beta_{1}$ & Premise \\ 
2. $\nnot \alpha_{2} \to \nnot \beta_{2}$ & Premise \\
3. $(\alpha_{1} \to \beta_{1}) \to (\nnot \nnot \alpha_{1} \to \nnot \nnot \beta_{1})$ & Ax17 \\
4. $\nnot \nnot \alpha_{1} \to \nnot \nnot \beta_{1}$ & MP, 1, 3 \\
5. $(\nnot \nnot \alpha_{1} \land \nnot \nnot \nnot \alpha_{2}) \to \nnot \nnot \alpha_{1}$ & Ax3 \\
6. $\overbrace{(\nnot \nnot \alpha_{1} \land \nnot \nnot \nnot \alpha_{2}) \to \nnot \nnot \beta_{1}}^{\varphi}$ & Lemma 1.2, 4, 5 \\
7. $(\nnot \alpha_{2} \to \nnot \beta_{2}) \to (\nnot \nnot \nnot \alpha_{2} \to \nnot \nnot \nnot \beta_{2})$ & Ax17 \\
8. $\nnot \nnot \nnot \alpha_{2} \to \nnot \nnot \nnot \beta_{2}$ & MP, 2, 7 \\
9. $(\nnot \nnot \alpha_{1} \land \nnot \nnot \nnot \alpha_{2}) \to \nnot \nnot \nnot \alpha_{2}$ & Ax4 \\
10. $\overbrace{(\nnot \nnot \alpha_{1} \land \nnot \nnot \nnot \alpha_{2}) \to \nnot \nnot \nnot \beta_{2}}^{\psi}$ & Lemma 1.2, 8, 9 \\
11. $\varphi \to (\psi \to ((\nnot \nnot \alpha_{1} \land \nnot \nnot \nnot \alpha_{2}) \to (\nnot \nnot \beta_{1} \land \nnot \nnot \nnot \beta_{2})))$ & Ax5 \\
12. $\psi \to ((\nnot \nnot \alpha_{1} \land \nnot \nnot \nnot \alpha_{2}) \to (\nnot \nnot \beta_{1} \land \nnot \nnot \nnot \beta_{2}))$ & MP, 6, 11 \\
13. $(\nnot \nnot \alpha_{1} \land \nnot \nnot \nnot \alpha_{2}) \to (\nnot \nnot \beta_{1} \land \nnot \nnot \nnot \beta_{2})$ & MP, 10, 12 \\
14. $\nnot \nnot (\alpha_{1} \land \nnot \alpha_{2}) \to (\nnot \nnot \alpha_{1} \land \nnot \nnot \nnot \alpha_{2})$ & Ax14 ($\to$) \\
15. $\nnot \nnot (\alpha_{1} \land \nnot \alpha_{2}) \to (\nnot \nnot \beta_{1} \land \nnot \nnot \nnot \beta_{2})$ & Lemma 1.2, 13, 14 \\
16. $(\nnot \nnot \beta_{1} \land \nnot \nnot \nnot \beta_{2}) \to \nnot \nnot (\beta_{1} \land \nnot \beta_{2})$ & Ax14 ($\leftarrow$) \\
17. $\nnot \nnot (\alpha_{1} \land \nnot \alpha_{2}) \to \nnot \nnot (\beta_{1} \land \nnot \beta_{2})$ & Lemma 1.2, 15, 16 \\
18. $\nnot (\alpha_{1} \to \alpha_{2}) \to \nnot \nnot (\alpha_{1} \land \nnot \alpha_{2})$ & Ax10 ($\to$) \\
19. $\nnot (\alpha_{1} \to \alpha_{2}) \to \nnot \nnot (\beta_{1} \land \nnot \beta_{2})$ & Lemma 1.2, 17, 18 \\
20. $\nnot \nnot (\beta_{1} \land \nnot \beta_{2}) \to \nnot (\beta_{1} \to \beta_{2})$ & Ax10 ($\leftarrow$) \\
21. $\nnot (\alpha_{1} \to \alpha_{2}) \to \nnot (\beta_{1} \to \beta_{2})$ & Lemma 1.2, 19, 20
\end{tabular} 
\end{table}

\end{proof}
\newpage
Having proved that our calculus is algebraizable in the sense Blok and Pigozzi,  we have a corres\-ponding equivalent algebraic semantics $\mathrm{Alg}^{*}(\mathbf{L}_{\mathrm{QN4}})$ defined as follows.

\begin{dfn} 
 An $\mathrm{Alg}^{*}(\mathbf{L}_{\mathrm{QN4}})$-algebra is a structure $\mathbf{A} = \langle A; \land, \lor, \to, \nnot \rangle$ which satisfies the following equations and quasi-equations:
\begin{enumerate}
\item $E(\alpha)$ for each $\alpha \in \mathbf{Ax}$.
\item $E(\Delta(\alpha, \alpha))$.
\item $E(\Delta(\alpha, \beta))$ implies $\alpha \approx \beta$.
\item $E(\alpha)$ and $E(\alpha \to \beta)$ implies $E(\beta)$.
\end{enumerate}
\label{algebra}
\end{dfn}
\noindent As an example of the notation $E(\alpha)$ above, for each axiom $\varphi \in \mathbf{Ax}$, the class of algebras $\mathrm{Alg}^{*}(\mathbf{L}_{\mathrm{QN4}})$ must satisfy $\varphi \approx \varphi \to \varphi$. Taking $\mathbf{Ax1}$ as an example, the class $\mathrm{Alg}^{*}(\mathbf{L}_{\mathrm{QN4}})$ has \linebreak $\alpha \to (\beta \to \alpha) \approx (\alpha \to (\beta \to \alpha))\to (\alpha \to (\beta \to \alpha))$ as one of its equations.

\section{$\mathrm{Alg}^*(\mathbf{L}_{\mathrm{QN4}}) = \mathrm{QN4}$}
\label{semantics}

In order to prove that the class of algebras introduced in 	Definition \ref{algebra} is term-equivalent to the class of QN4-lattices
(Definition \ref{qn4l}), that is, $\mathrm{Alg}^*(\mathbf{L}_{\mathrm{QN4}}) = \mathrm{QN4}$, we have to prove that $\mathrm{Alg}^*(\mathbf{L}_{\mathrm{QN4}})$ satisfies all equations that axiomatize $\mathrm{QN4}$ and that $\mathrm{QN4}$ satisfies the equations (Definition \ref{algebra}.1 and Definition \ref{algebra}.2) and quasi-equations (Definition \ref{algebra}.3 and Definition \ref{algebra}.4) that axiomatize $\mathrm{Alg}^*(\mathbf{L}_{\mathrm{QN4}})$. This is the content of the next Proposition.

\begin{prop}
$\mathrm{Alg}^*(\mathbf{L}_{\mathrm{QN4}}) \subseteq \mathrm{QN4}$.
\label{ida}
\end{prop}

\begin{proof} 
For proving \textbf{QN4a}, we need to show that the idempotent, commutative, absorption, associative and distributive laws holds for every $\mathbf{A} \in \mathrm{Alg}^*(\mathbf{L}_{\mathrm{QN4}})$. 

\begin{enumerate}

\item Idempotent laws. 

For the law $x \land x = x$, we need to have that $(x \land x) \to x = |(x \land x) \to x|$, $x \to (x \land x) = |x \to (x \land x)|$, $\nnot (x \land x) \to \nnot x = |\nnot (x \land x) \to \nnot x|$ and $\nnot x \to \nnot (x \land x) = |\nnot x \to \nnot (x \land x)|$. In order to have these four equations in the algebra, we must prove in the logic the following four axioms:

\begin{enumerate}
\item $(\alpha \land \alpha) \to \alpha$, this is an instatiation of \textbf{Ax3}.

\item $\alpha \to (\alpha \land \alpha)$, is demonstrated using \textbf{Ax5}, Lemma 1.1 and MP.

\item $\nnot (\alpha \land \alpha) \to \nnot \alpha$, is demonstrated using \textbf{Ax20}, Lemma 1.1 and MP.

\item $\nnot \alpha \to \nnot (\alpha \land \alpha)$, this is an instantiation of \textbf{Ax18}.
\end{enumerate}

The same idea applies to $x \lor x = x$.

\item Commutative laws

For the law $x \land y = y \land x$, we have: 

\begin{enumerate}
\item $(\alpha \land \beta) \to (\beta \land \alpha)$
\begin{table}[!ht]
\centering
\begin{tabular}{ll}
1. $((\alpha \land \beta) \to \beta) \to (((\alpha \land \beta) \to \alpha) \to ((\alpha \land \beta) \to (\beta \land \alpha)))$ & Ax5 \\
2. $(\alpha \land \beta) \to \beta$ & Ax4 \\
3. $((\alpha \land \beta) \to \alpha) \to ((\alpha \land \beta) \to (\beta \land \alpha))$ & MP, 1, 2 \\
4. $(\alpha \land \beta) \to \alpha$ & Ax3 \\
5. $(\alpha \land \beta) \to (\beta \land \alpha)$ & MP, 3, 4
\end{tabular} 
\end{table} 

\item $(\beta \land \alpha) \to (\alpha \land \beta)$, this is an instantiation of previous item.

\item $\nnot (\alpha \land \beta) \to \nnot (\beta \land \alpha)$, this is \textbf{Ax19}.

\item $\nnot (\beta \land \alpha) \to \nnot (\alpha \land \beta)$, this is an instantiation of \textbf{Ax19}.
\end{enumerate}

The same idea applies to $x \lor y = y \lor x$.






\item Absorption laws.

For the law $x \land (x \lor y) = x$, we have:

\begin{enumerate}
\item $(\alpha \land (\alpha \lor \beta)) \to \alpha$, this is an instantiation of \textbf{Ax3}.

\item $\alpha \to (\alpha \land (\alpha \lor \beta))$
\begin{table}[!ht]
\centering
\begin{tabular}{ll}
1. $(\alpha \to \alpha) \to ((\alpha \to (\alpha \lor \beta)) \to (\alpha \to (\alpha \land (\alpha \lor \beta))))$ & Ax5 \\
2. $\alpha \to \alpha$ & Lemma 1.1 \\
3. $(\alpha \to (\alpha \lor \beta)) \to (\alpha \to (\alpha \land (\alpha \lor \beta)))$ & MP, 1, 2 \\
4. $\alpha \to (\alpha \lor \beta)$ & Ax6 \\
5. $\alpha \to (\alpha \land (\alpha \lor \beta))$ & MP, 3, 4
\end{tabular} 
\end{table} 

\newpage

\item $\nnot (\alpha \land (\alpha \lor \beta)) \to \nnot \alpha$
\begin{table}[!ht]
\centering
\begin{tabular}{ll}
1. $\nnot (\alpha \land (\alpha \lor \beta)) \to \nnot ((\alpha \land \alpha) \lor (\alpha \land \beta))$ & Ax12 ($\to$) \\
2. $\nnot ((\alpha \land \alpha) \lor (\alpha \land \beta)) \to (\nnot (\alpha \land \alpha) \land \nnot (\alpha \land \beta))$ & Ax9 ($\to$) \\
3. $\nnot (\alpha \land (\alpha \lor \beta)) \to (\nnot (\alpha \land \alpha) \land \nnot (\alpha \land \beta))$ & Lemma 1.2, 1, 2 \\
4. $(\nnot (\alpha \land \alpha) \land \nnot (\alpha \land \beta)) \to \nnot (\alpha \land \alpha)$ & Ax3 \\
5. $\nnot (\alpha \land (\alpha \lor \beta)) \to \nnot (\alpha \land \alpha)$ & Lemma 1.2, 3, 4 \\
6. $(\nnot \alpha \to \nnot \alpha) \to (\nnot (\alpha \land \alpha) \to \nnot \alpha)$ & Ax20 \\
7. $\nnot \alpha \to \nnot \alpha$ & Lemma 1.1 \\
8. $\nnot (\alpha \land \alpha) \to \nnot \alpha$ & MP, 6, 7 \\
9. $\nnot (\alpha \land (\alpha \lor \beta)) \to \nnot \alpha$ & Lemma 1.2, 5, 8
\end{tabular} 
\end{table} 

\item $\nnot \alpha \to \nnot (\alpha \land (\alpha \lor \beta))$, this is an instantiation of \textbf{Ax18}.

\end{enumerate}

The same idea applies to $x \lor (x \land y) = x$.






\item Associative laws.

For the law $x \land (y \land z) = (x \land y) \land z$, we have:

\begin{enumerate}

\item $(\alpha \land (\beta \land \gamma)) \to ((\alpha \land \beta) \land \gamma)$
\begin{table}[!ht]
\centering
\begin{tabular}{ll}
1. $((\alpha \land (\beta \land \gamma)) \to (\alpha \land \beta)) \to ((\alpha \land (\beta \land \gamma)) \to \gamma) \to ((\alpha \land (\beta \land \gamma)) \to ((\alpha \land \beta) \land \gamma)))$ & Ax5 \\
2. $(\alpha \land (\beta \land \gamma)) \to \alpha$ & Ax3 \\
3. $(\alpha \land (\beta \land \gamma)) \to \beta \land \gamma$ & Ax4 \\
4. $(\beta \land \gamma) \to \beta$ & Ax3 \\
5. $(\alpha \land (\beta \land \gamma)) \to \beta$ & Lemma 1.2, 3, 4 \\
6. $((\alpha \land (\beta \land \gamma)) \to \alpha) \to (((\alpha \land (\beta \land \gamma) \to \beta) \to (((\alpha \land (\beta \land \gamma)) \to (\alpha \land \beta)))$ & Ax5 \\
7. $((\alpha \land (\beta \land \gamma) \to \beta) \to (((\alpha \land (\beta \land \gamma)) \to (\alpha \land \beta))$ & MP, 2, 6 \\
8. $((\alpha \land (\beta \land \gamma)) \to (\alpha \land \beta)$ & MP, 5, 7 \\
9. $(\alpha \land (\beta \land \gamma)) \to \gamma) \to ((\alpha \land (\beta \land \gamma)) \to ((\alpha \land \beta) \land \gamma))$ & MP, 1, 8 \\
10. $(\beta \land \gamma) \to \gamma$ & Ax4 \\
11. $(\alpha \land (\beta \land \gamma)) \to \gamma$ & Lemma 1.2, 3, 10 \\
12. $(\alpha \land (\beta \land \gamma)) \to ((\alpha \land \beta) \land \gamma)$ & MP, 9, 11
\end{tabular} 
\end{table} 

\item $((\alpha \land \beta) \land \gamma) \to (\alpha \land (\beta \land \gamma))$
\begin{table}[!ht]
\centering
\begin{tabular}{ll}
1. $(((\alpha \land \beta) \land \gamma) \to \alpha) \to (((\alpha \land \beta) \land \gamma) \to (\beta \land \gamma)) \to (((\alpha \land \beta) \land \gamma) \to (\alpha \land (\beta \land \gamma))))$ & Ax5 \\
2. $((\alpha \land \beta) \land \gamma) \to (\alpha \land \beta)$ & Ax3 \\
3. $(\alpha \land \beta) \to \alpha$ & Ax3 \\
4. $((\alpha \land \beta) \land \gamma) \to \alpha$ & Lemma 1.2, 2, 3 \\
5. $((\alpha \land \beta) \land \gamma) \to (\beta \land \gamma)) \to (((\alpha \land \beta) \land \gamma) \to (\alpha \land (\beta \land \gamma)))$ & MP, 1, 4 \\
6. $((\alpha \land \beta) \land \gamma) \to \gamma$ & Ax4 \\
7. $(\alpha \land \beta) \to \beta$ & Ax4 \\
8. $((\alpha \land \beta) \land \gamma) \to \beta$ & Lemma 1.2, 2, 7 \\
9. $(((\alpha \land \beta) \land \gamma) \to \beta) \to (((\alpha \land \beta) \land \gamma) \to \gamma) \to (((\alpha \land \beta) \land \gamma) \to (\beta \land \gamma)))$ & Ax5 \\ 
10. $((\alpha \land \beta) \land \gamma) \to \gamma) \to (((\alpha \land \beta) \land \gamma) \to (\beta \land \gamma))$ & MP, 8, 9 \\
11. $((\alpha \land \beta) \land \gamma) \to (\beta \land \gamma)$ & MP, 6, 10 \\
12. $((\alpha \land \beta) \land \gamma) \to (\alpha \land (\beta \land \gamma))$ & MP 5, 11
\end{tabular} 
\end{table} 

\noindent \item $\nnot (\alpha \land (\beta \land \gamma)) \to \nnot ((\alpha \land \beta) \land \gamma)$, this is \textbf{Ax11} $(\to)$.

\item $\nnot ((\alpha \land \beta) \land \gamma) \to \nnot (\alpha \land (\beta \land \gamma))$, this is \textbf{Ax11} $(\leftarrow)$.

\end{enumerate}

The same idea applies to  $x \lor (y \lor z) = (x \lor y) \lor z$.






\item Distributive laws. 

Axioms $\mathbf{Ax1}$-$\mathbf{Ax8}$ of $\mathbf{L}_{\mathbf{QN4}}$ are the axioms of the Positive Logic and it is known that the distributive law holds in this logic. Distributive law and $\mathbf{Ax11}$ give us the distributivity in the lattice.


\centering

\end{enumerate}

Clearly, \textbf{QN4d} is axiom 10, \textbf{QN4e.1} is axiom 15, \textbf{QN4e.3} is axiom 9 and \textbf{QN4e.4} is axiom 14. For \textbf{QN4e.2}, that is, $\nnot a = \nnot \nnot \nnot a$, we have that $\nnot \nnot \nnot a \leq \nnot a$ by axiom 22. It remains to prove that $\nnot a \leq \nnot \nnot \nnot a$, this is an instantiation of axiom 15. Instead proving \textbf{QN4b} and \textbf{QN4c}, we can prove that $\mathrm{Alg}^*(\mathbf{L}_{\mathrm{QN4}})$ satisfies the equations of Proposition \ref{variety} and these proves are straightforward.

\end{proof}

\begin{prop}
$\mathrm{QN4} \subseteq \mathrm{Alg}^*(\mathbf{L}_{\mathrm{QN4}})$.
\end{prop}

\begin{proof}
Let $\mathbf{A} \in \mathrm{QN4}$, and let $a, b, c \in A$ be generic elements.
By Theorem \ref{thmrep}, we assume that $\mathbf{A}$ is a twist-structure, and from now on we also denote $a = \langle a_1, a_2 \rangle, b = \langle b_1, b_2 \rangle, c = \langle c_1, c_2 \rangle $. Note that, proving $E(\alpha)$ for a given term $\alpha$ is equivalent to showing that $\pi_1(\alpha) = 1$. We shall use this observation without further notice throughout the proof.

It is very easy to see that the twist-structure definitions, together with the Brouwerian algebra pro\-perties, entail that $\pi_1(\mathbf{Axn}) = 1$  for $1 \leq n \leq 8$. In the case of $E(a \eq b)$, it is equivalent to prove that $\pi_1 (a \eq b) = 1$, which in turn is equivalent to proving $\pi_1 (a) = \pi_1 (b ) $, this is, $a_1 = b_1$. So,

\begin{itemize}

\item $E(\mathbf{Ax9}) = E(\nnot(a \lor b) \eq (\nnot a \land \nnot b))$

$\pi_1[\nnot(a \lor b)] = \pi_1[\nnot (\langle a_1, a_2 \rangle \lor \langle b_1, b_2 \rangle)] = \pi_1[\nnot \langle a_1 \lor b_1, a_2 \land b_2 \rangle] = \pi_1[\langle a_2 \land b_2, \Box(a_1 \lor b_1) \rangle] = a_2 \land b_2$. 

On the other hand, $\pi_1[\nnot a \land \nnot b] = \pi_1[\nnot \langle a_1, a_2 \rangle \land \nnot \langle b_1, b_2 \rangle] = \pi_1[\langle a_2, \Box a_1 \rangle \land \langle b_2, \Box b_1 \rangle] = \pi_1[\langle a_2 \land b_2, \Box(\Box a_1 \lor \Box b_1) \rangle] = a_2 \land b_2$.

\item $E(\mathbf{Ax10}) = E(\nnot (a \to b) \eq \nnot \nnot (a \land \nnot b))$

$\pi_1[\nnot (a \to b)] = \pi_1[\nnot (\langle a_1, a_2 \rangle \to \langle b_1, b_2 \rangle)] = \pi_1[\nnot \langle a_1 \to b_1, \Box a_1 \land b_2 \rangle] = \pi_1[\langle \Box a_1 \land b_2, \Box(a_1 \to b_1) \rangle] = \Box a_1 \land b_2$.

In contrast, $\pi_1[\nnot \nnot (a \land \nnot b)] = \pi_1[\nnot \nnot (\langle a_1, a_2 \rangle \land \nnot \langle b_1, b_2 \rangle)] = \pi_1[\nnot \nnot (\langle a_1, a_2 \rangle \land \langle b_2, \Box b_1 \rangle)] = \pi_1[\nnot \nnot \langle a_1 \land b_2, \Box(a_2 \lor \Box b_1) \rangle] = \pi_1[\nnot \langle \Box(a_2 \lor \Box b_1), \Box(a_1 \land b_2) \rangle] = \pi_1[\langle \Box(a_1 \land b_2), \Box(\Box(a_2 \lor \Box b_1)) \rangle] = \Box(a_1 \land b_2) = \Box a_1 \land \Box b_2 = \Box a_1 \land b_2$.

\item $E(\mathbf{Ax11}) = E(\nnot(a \land (b \land c)) \eq \nnot((a \land b) \land c))$

$\pi_1[\nnot(a \land (b \land c))] = \pi_1[\nnot(\langle a_1, a_2 \rangle \land (\langle b_1, b_2 \rangle \land \langle c_1, c_2 \rangle))] = \pi_1[\nnot(\langle a_1, a_2 \rangle \land \langle b_1 \land c_1, \Box(b_2 \lor c_2) \rangle)] = \pi_1[\nnot \langle a_1 \land (b_1 \land c_1), \Box(a_2 \lor \Box(b_2 \lor c_2)) \rangle] = \pi_1[\langle \Box(a_2 \lor \Box(b_2 \lor c_2)), \Box(a_1 \land (b_1 \land c_1)) \rangle] = \Box(a_2 \lor \Box(b_2 \lor c_2)) = \Box(a_2 \lor (b_2 \lor c_2)) = a_2 \lor (b_2 \lor c_2) = a_2 \lor b_2 \lor c_2$.

However, $\pi_1[\nnot((a \land b) \land c)] = \pi_1[\nnot((\langle a_1, a_2 \rangle \land \langle b_1, b_2 \rangle) \land \langle c_1, c_2 \rangle)] = \pi_1[\nnot(\langle a_1 \land b_1, \Box(a_2 \lor b_2) \rangle \land \langle c_1, c_2 \rangle)] = \pi_1[\nnot \langle a_1 \land b_1 \land c_1, \Box(\Box(a_2 \lor b_2) \lor c_2) \rangle] = \pi_1[\langle \Box(\Box(a_2 \lor b_2) \lor c_2), \Box((a_1 \land b_1 \land c_1)) \rangle] = \Box(\Box(a_2 \lor b_2) \lor c_2) = \Box((a_2 \lor b_2) \lor c_2) = (a_2 \lor b_2) \lor c_2 = a_2 \lor b_2 \lor c_2$.

\item $E(\mathbf{Ax12}) = E(\nnot(a \land (b \lor c)) \eq \nnot((a \land b) \lor (a \land c)))$

$\pi_1[\nnot(a \land (b \lor c))] = \pi_1[\nnot(\langle a_1, a_2 \rangle \land (\langle b_1, b_2 \rangle \lor \langle c_1, c_2 \rangle))] \pi_1[\nnot(\langle a_1, a_2 \rangle \land \langle b_1 \lor c_1, b_2 \land c_2 \rangle)] = \pi_1[\nnot \langle a_1 \land (b_1 \lor c_1), \Box(a_2 \lor (b_2 \land c_2)) \rangle] = \pi_1[\langle \Box(a_2 \lor (b_2 \land c_2)), \Box(a_1 \land (b_1 \lor c_1)) \rangle] = \Box(a_2 \lor (b_2 \land c_2)) = a_2 \lor (b_2 \land c_2)$.

On the other hand, $\pi_1[\nnot((a \land b) \lor (a \land c))] = \pi_1[\nnot((\langle a_1, a_2 \rangle \land \langle b_1, b_2 \rangle) \lor (\langle a_1, a_2 \rangle \land \langle c_1, c_2 \rangle))] = \pi_1[\nnot( \langle a_1 \land b_1, \Box(a_2 \lor b_2) \rangle \lor \langle a_1 \land c_1, \Box(a_2 \lor c_2) \rangle)] = \pi_1[\nnot \langle (a_1 \land b_1) \lor (a_1 \land c_1), \Box(a_2 \lor b_2) \land \Box(a_2 \lor c_2) \rangle] = \pi_1[\langle \Box(a_2 \lor b_2) \land \Box(a_2 \lor c_2), \Box((a_1 \land b_1) \lor (a_1 \land c_1)) \rangle] = \Box(a_2 \lor b_2) \land \Box(a_2 \lor c_2) = (a_2 \lor b_2) \land (a_2 \lor c_2) = a_2 \lor (b_2 \land c_2)$.

\item $E(\mathbf{Ax13}) = E(\nnot(a \lor (b \land c)) \eq \nnot((a \lor b) \land (a \lor c)))$

$\pi_1[\nnot(a \lor (b \land c))] = \pi_1[\nnot(\langle a_1, a_2 \rangle \lor (\langle b_1, b_2 \rangle \land \langle c_1, c_2 \rangle)] = \pi_1[\nnot( \langle a_1, a_2 \rangle \lor \langle b_1 \land c_1, \Box(b_2 \lor c_2) \rangle)] = \pi_1[\nnot \langle a_1 \lor (b_1 \land c_1), a_2 \land \Box(b_2 \lor c_2) \rangle] = \pi_1[\langle a_2 \land \Box(b_2 \lor c_2), \Box(a_1 \lor (b_1 \land c_1)) \rangle] = a_2 \land \Box(b_2 \lor c_2) = a_2 \land (b_2 \lor c_2)$.

However, $\pi_1[\nnot((a \lor b) \land (a \lor c))] = \pi_1[\nnot((\langle a_1, a_2 \rangle \lor \langle b_1, b_2 \rangle) \land (\langle a_1, a_2 \rangle \lor \langle c_1, c_2 \rangle))] = \pi_1[\nnot(\langle a_1 \lor b_1, a_2 \land b_2 \rangle \land \langle a_1 \lor c_1, a_2 \land c_2 \rangle)] = \pi_1[\nnot \langle (a_1 \lor b_1) \land (a_1 \lor c_1), \Box((a_2 \land b_2) \lor (a_2 \land c_2)) \rangle] = \pi_1[\langle \Box((a_2 \land b_2) \lor (a_2 \land c_2)), \Box((a_1 \lor b_1) \land (a_1 \lor c_1)) \rangle] = \Box((a_2 \land b_2) \lor (a_2 \land c_2)) = (a_2 \land b_2) \lor (a_2 \land c_2) = a_2 \land (b_2 \lor c_2)$.

\item $E(\mathbf{Ax14}) = E(\nnot \nnot(a \land b) \eq (\nnot \nnot a \land \nnot \nnot b))$

$\pi_1[\nnot \nnot(a \land b)] = \pi_1[\nnot \nnot(\langle a_1, a_2 \rangle \land \langle b_1, b_2 \rangle)] = \pi_1[\nnot \nnot \langle a_1 \land b_1, \Box(a_2 \lor b_2) \rangle] = \pi_1[\nnot \langle \Box(a_2 \lor b_2), \Box(a_1 \land b_1) \rangle] = \pi_1[\langle \Box(a_1 \land b_1), \Box(\Box(a_2 \lor b_2)) \rangle] = \Box(a_1 \land b_1) = \Box a_1 \land \Box b_1$.

In contrast, $\pi_1[\nnot \nnot a \land \nnot \nnot b] = \pi_1[\nnot \nnot \langle a_1, a_2 \rangle \land \nnot \nnot \langle b_1, b_2 \rangle] = \pi_1[\nnot \langle a_2, \Box a_1 \rangle \land \nnot \langle b_2, \Box b_1 \rangle] = \pi_1[\langle \Box a_1, \Box a_2 \rangle \land \langle \Box b_1, \Box b_2 \rangle] = \pi_1[\langle \Box a_1 \land \Box b_1, \Box(\Box a_2 \lor \Box b_2) \rangle] = \Box a_1 \land \Box b_1$.

Already in case of $E(a \to b)$ saying this is equivalent to proving that $\pi_1 (a) \leq \pi_1 (b) $, this is, $a_1 \leq b_1$.

\item $E(\mathbf{Ax15}) = E(a \to \nnot \nnot a)$

$\pi_1[a] = \pi_1[\langle a_1, a_2 \rangle] = a_1$.

On the other hand, $\pi_1[\nnot \nnot a] = \pi_1[\nnot \nnot \langle a_1, a_2 \rangle] = \pi_1[\nnot \langle a_2, \Box a_1 \rangle] = \pi_1[\langle \Box a_1, \Box a_2 \rangle] = \Box a_1$.

\item $E(\mathbf{Ax16}) = E(a \to (\nnot a \to \nnot (a \to a)))$

$\pi_1[a] = \pi_1[\langle a_1, a_2 \rangle] = a_1$.

However, $\pi_1[\nnot a \to \nnot(a \to a)] = \pi_1[\nnot \langle a_1, a_2 \rangle \to \nnot(\langle a_1, a_2 \rangle \to \langle a_1, a_2 \rangle)] = \pi_1[\langle a_2, \Box a_1 \rangle \to \nnot \langle a_1 \to a_1, \Box a_1 \land a_2 \rangle] = \pi_1[\langle a_2, \Box a_1 \rangle \to \langle \Box a_1 \land a_2, \Box(a_1 \to a_1) \rangle = \pi_1[\langle a_2 \to (\Box a_1 \land a_2), \Box a_2 \land \Box(a_1 \to a_1)] = a_2 \to (\Box a_1 \land a_2)$.

\item $E(\mathbf{Ax17}) = E((a \to b) \to (\nnot \nnot a \to \nnot \nnot b))$

$\pi_1[a \to b] = \pi_1[\langle a_1, a_2 \rangle \to \langle b_1, b_2 \rangle] = \pi_1[\langle a_1 \to b_1, \Box a_1 \land b_2 \rangle] = a_1 \to b_1$.

In contrast, $\pi_1[\nnot \nnot a \to \nnot \nnot b] = \pi_1[\nnot \nnot \langle a_1, a_2 \rangle \to \nnot \nnot \langle b_1, b_2 \rangle] = \pi_1[\nnot \langle a_2, \Box a_1 \rangle \to \nnot \langle b_2, \Box b_1 \rangle] = \pi_1[\langle \Box a_1, \Box a_2 \rangle \to \langle \Box b_1, \Box b_2 \rangle] = \pi_1[\langle \Box a_1 \to \Box b_1, \Box \Box a_1 \land \Box b_2 \rangle] = \Box a_1 \to \Box b_1$.

\item $E(\mathbf{Ax18}) = E(\nnot a \to \nnot(a \land b))$

$\pi_1[\nnot a] = \pi_1[\nnot \langle a_1, a_2 \rangle] = \pi_1[\langle a_2, \Box a_1 \rangle] = a_2$.

On the other hand, $\pi_1[\nnot(a \land b)] = \pi_1[\nnot(\langle a_1, a_2 \rangle \land \langle b_1, b_2 \rangle)] = \pi_1[\nnot \langle a_1 \land b_1, \Box(a_2 \lor b_2) \rangle] = \pi_1[\langle \Box(a_2 \lor b_2), \Box(a_1 \land b_1) \rangle] = \Box(a_2 \lor b_2) = a_2 \lor b_2$.

\item $E(\mathbf{Ax19}) = E(\nnot(a \land b) \to \nnot(b \land a))$

$\pi_1[\nnot(a \land b)] = \pi_1[\nnot(\langle a_1, a_2 \rangle \land \langle b_1, b_2 \rangle)] = \pi_1[\nnot \langle a_1 \land b_1, \Box(a_2 \lor b_2) \rangle] = \pi_1[\langle \Box(a_2 \lor b_2), \Box(a_1 \land b_1) \rangle] = \Box(a_2 \lor b_2) = a_2 \lor b_2$.

However, $\pi_1[\nnot(b \land a)] = \pi_1[\nnot(\langle b_1, b_2 \rangle \land \langle a_1, a_2 \rangle)] = \pi_1[\nnot \langle b_1 \land a_1, \Box(b_2 \lor a_2) \rangle] = \pi_1[\langle \Box(b_2 \lor a_2), \Box(b_1 \land a_1) \rangle] = \Box(b_2 \lor a_2) = b_2 \lor a_2$.

\item $E(\mathbf{Ax20}) = E((\nnot a \to \nnot b) \to (\nnot (a \land b) \to \nnot b))$

$\pi_1[\nnot a \to \nnot b] = \pi_1[\nnot \langle a_1, a_2 \rangle \to \nnot \langle b_1, b_2 \rangle] = \pi_1[\langle a_2, \Box a_1 \rangle \to \langle b_2, \Box b_1 \rangle] = \pi_1[\langle a_2 \to b_2, \Box a_2 \land \Box b_1 \rangle] = a_2 \to b_2$.

In contrast, $\pi_1[\nnot(a \land b) \to \nnot b] = \pi_1[\nnot(\langle a_1, a_2 \rangle \land \langle b_1, b_2 \rangle) \to \nnot \langle b_1, b_2 \rangle] = \pi_1[\nnot \langle a_1 \land b_1, \Box(a_2 \lor b_2) \rangle \to \langle b_2, \Box b_1 \rangle] = \pi_1[\langle \Box(a_2 \lor b_2), \Box(a_1 \land b_1) \rangle \to \langle b_2, \Box b_1 \rangle] = \pi_1[\langle \Box(a_2 \lor b_2) \to b_2, \Box\Box(a_2 \lor b_2) \land \Box b_1 \rangle] = \Box(a_2 \lor b_2) \to b_2 = (a_2 \lor b_2) \to b_2$.

\item $E(\mathbf{Ax21}) = E((\nnot a \to \nnot b) \to ((\nnot c \to \nnot \theta) \to (\nnot(a \land c) \to \nnot(b \land \theta))))$

$\pi_1[\nnot a \to \nnot b] = \pi_1[\nnot \langle a_1, a_2 \rangle \to \nnot \langle b_1, b_2 \rangle] = \pi_1[\langle a_2, \Box a_1 \rangle \to \langle b_2, \Box b_1 \rangle] = \pi_1[\langle a_2 \to b_2, \Box a_2 \land \Box b_1 \rangle] = a_2 \to b_2$.

On the other hand, $\pi_1[(\nnot c \to \nnot d) \to (\nnot(a \land c) \to \nnot(b \land d))] = \pi_1[(\nnot \langle c_1, c_2 \rangle \to \nnot \langle d_1, d_2 \rangle) \to (\nnot(\langle a_1, a_2 \rangle \land \langle c_1, c_2 \rangle) \to \nnot (\langle b_1, b_2 \rangle \land \langle d_1, d_2 \rangle)] = \pi_1[(\langle c_2, \Box c_1 \rangle \to \langle d_2, \Box d_1 \rangle) \to (\nnot \langle a_1 \land c_1, \Box(a_2 \lor c_2) \rangle \to \nnot \langle b_1 \land d_1, \Box(b_2 \lor d_2) \rangle)] = \pi_1[\langle c_2 \to d_2, \Box c_2 \land \Box d_1 \rangle \to (\langle \Box(a_2 \lor c_2), \Box(a_1 \land c_1) \rangle \to \langle \Box(b_2 \lor d_2), \Box(b_1 \land d_1) \rangle)] = \pi_1[\langle c_2 \to d_2, \Box c_2 \land \Box d_1 \rangle \to \langle \Box(a_2 \lor c_2) \to \Box(b_2 \lor d_2), \Box \Box(a_2 \lor c_2) \land \Box(b_1 \land d_1) \rangle] = \pi_1[\langle (c_2 \to d_2) \to (\Box(a_2 \lor c_2) \to \Box(b_2 \lor d_2)), \Box(c_2 \to d_2) \land (\Box \Box(a_2 \lor c_2) \land \Box(b_1 \land d_1)) \rangle] = (c_2 \to d_2) \to (\Box(a_2 \lor c_2) \to \Box(b_2 \lor d_2)) = (c_2 \to d_2) \to ((a_2 \lor c_2) \to (b_2 \lor d_2))$.

\item $E(\mathbf{Ax22}) = E(\nnot \nnot \nnot a \to \nnot a)$

$\pi_1[\nnot \nnot \nnot a] = \pi_1[\nnot \nnot \nnot \langle a_1, a_2 \rangle] = \pi_1[\nnot \nnot \langle a_2, \Box a_1 \rangle] = \pi_1[\nnot \langle \Box a_1, \Box a_2 \rangle] = \pi_1[\langle \Box a_2, \Box \Box a_1 \rangle] = \Box a_2 = a_2$.

However, $\pi_1[\nnot a] = \pi_1[\nnot \langle a_1, a_2 \rangle] = \pi_1[\langle a_2, \Box a_1 \rangle] = a_2$.
\end{itemize}

In \textbf{5.2}, we have to prove that $a \to a = (a \to a) \to (a \to a)$ and that $\nnot a \to \nnot a = (\nnot a \to \nnot a) \to (\nnot a \to \nnot a)$. Taking $|x| = y = (a \to a)$ in Proposition \ref{variety}.1, we have that $(a \to a) \to (a\to a) = a \to a$, that is what we wanted to prove. The same idea for negation.

In \textbf{5.3}, we have to prove that if $ a = a \to a$, $a \to b= (a \to b)\to(a \to b)$ then $b = b \to b$. Again, using Proposition \ref{variety}.1, taking $|x| = a \to a$ and $y = b$, we have that $(a \to a) \to b = b$, but we have that $a \to a = a$ and therefore $a \to b = b$, but as $a \to b= (a \to b)\to(a \to b)$ and  $a \to b = b$ we have that $b = b \to b$.

In \textbf{5.4}, we want to prove that if $a \to b = | a \to b |$, $b \to a = | b \to a |$, $\nnot a \to \nnot a = | \nnot a \to \nnot b |$, $\nnot b \to \nnot a = | \nnot b \to \nnot a |$, then $a = b$. As $a \to b = | a \to b |$ and $\nnot b \to \nnot a = | \nnot b \to \nnot a |$, we have $a \preceq b$ and  $\nnot b \preceq \nnot a$ and therefore by QN4c we conclude that $a \leq b$. We also have that $b \to a = | b \to a |$ and $\nnot a \to \nnot b = | \nnot a \to \nnot b |$ and therefore $b \preceq a$ and  $\nnot a \preceq \nnot b$ and again by QN4c we conclude that $b \leq a$. As $a \leq b$ and $b \leq a$ we have $a = b$ and this is what we wanted to prove.
\end{proof}

\begin{corollary} The class of $\mathrm{QN4}$-lattices and the class of $\mathrm{Alg}^*(\mathbf{L}_{\mathrm{QN4}})$-algebras coincide.
	\end{corollary}

\section{Future Work} 
\label{future}

By way of conclusion, we mention below a few potential directions for future research.

1. As shown by Spinks and Veroff \cite{Spinks2018}, N4-lattices may be axiomatized in the language having (besides the lattice connectives) only the negation ($\nnot$), the monoid conjunction ($*$) and the strong implication ($\Rightarrow$) as primitive;  this observation entails that the logic of N4-lattices can be presented as a contraction-free relevance logic, and may therefore  be more easily compared with other members of this wide family of non-classical logics. \cite[Remark~3.5]{Rivieccio2022} suggests that, similarly, the class of QN4-lattices may also be axiomatized in the language $\{\land, \lor, *, \Rightarrow, \nnot \}$. Based on this, we conjecture that it must be possible to give an alternative axiomatization that allows one to view $\mathbf{L}_{\mathrm{QN4}}$  as a logic related to the ones in the relevance family. 

2. The algebraizability result established in the present paper entails that finitary extensions of $\mathbf{L}_{\mathrm{QN4}}$ are in one-to-one correspondence with sub-quasi-varieties of QN4-lattices, with the axiomatic extensions corresponding to sub-varieties. The issue of sub(quasi)varieties of QN4-lattices is barely mentioned in~\cite{Rivieccio2022} -- see for instance~\cite[Theorem~4.7]{Rivieccio2022}. We believe that a systematic study of sub-quasi-varieties of QN4-lattices (i.e.~of finitary extensions of $\mathbf{L}_{\mathrm{QN4}}$) would be an interesting project to be explored in the course of future research. 
 
3. The twist construction, which is well known to be fundamental in the study of (quasi-)Nelson algebras and of the corresponding logic, has been  used in a series of recent papers~\cite{Rivieccio2020,Nascimento2021,Rivieccio2021a,Rivieccio2021b,Rivieccio2021c} to cha\-racterize a number of subreducts of (quasi-)Nelson algebras (corresponding to propositional fragments of (quasi-)Nelson logic). We believe it would be interesting to investigate to which extent the study developed in the above-mentioned papers can be extended to the case of  quasi-N4-lattices and $\mathbf{L}_{\mathrm{QN4}}$, and in particular to address the question of whether the lack of the specific  properties of quasi-Nelson algebras (integrality, the presence of the lattice constants) create any significant  technical obstacle  with regards to this issue. 

\bibliographystyle{eptcs}

\bibliography{references}

\end{document}